\documentclass[12pt]{article}
\usepackage{longtable,supertabular}
\usepackage{amsmath}
\usepackage{amssymb}
\usepackage{latexsym}
\usepackage{cite,a4wide,color}
\usepackage{multirow}
\usepackage{ntheorem}
\usepackage[square,numbers]{natbib}

\newtheorem{theorem}{Theorem}[section]

\newtheorem{lemma}{Lemma}[section]

\newtheorem*{proof}{Proof.}

\newtheorem{example}{Example}[section]

\theoremstyle{remark}

\newcommand{\m}{\mathrm{m}}
\newcommand{\ord}{\mathrm{ord}}

\newcommand{\bC}{{\mathbf{{C}}}}
\newcommand{\bF}{{\mathbf{{F}}}}

\newcommand{\bZ}{{\mathbf{{Z}}}}

\newcommand{\bc}{{\mathbf{c}}}

\newcommand{\bu}{{\mathbf{u}}}
\newcommand{\bv}{{\mathbf{v}}}

\newcommand{\lcm}{{\mathrm{lcm}}}

\newcommand{\Plotkin}{{\mathrm{Plotkin}}}

\title{\bf Repeated-Root Cyclic Codes with Optimal Parameters or Best Parameters Known\thanks{
The research of Hao Chen was supported by NSFC Grant 62032009. The research of Cunsheng Ding was supported by The Hong Kong Research Grants Council under Grant No. 16301123 and partially by
the UAEU-AUA joint research grant G00004614.
}
}

\author{Hao Chen, Conghui Xie, \thanks{Hao Chen and Conghui Xie are with the College of Information Science and Technology/Cyber Security, Jinan University, Guangzhou, Guangdong Province, 510632, China (e-mail: haochen@jnu.edu.cn, conghui@stu21.jnu.edu.cn).
}
and
Cunsheng Ding\thanks{
C. Ding is with the Department of Computer Science and Engineering, The Hong Kong University of Science and Technology, Hong Kong, China (cding@ust.hk).
}
}

\begin{document}
	
	\maketitle
	\begin{abstract}
	
Cyclic codes are the most studied subclass of linear codes and widely used in data storage and
communication systems.   Many cyclic codes have optimal parameters or the best parameters known.
They are divided into simple-root cyclic codes and repeated-root cyclic codes.  Although there are a huge
number of references on cyclic codes,  few of them are on repeated-root cyclic codes.  Hence, repeated-root cyclic codes are rarely studied.  There are a few families of distance-optimal repeated-root binary and
$p$-ary cyclic codes  for odd prime $p$ in the literature. However,
it is open whether there exists an infinite family of distance-optimal repeated-root cyclic codes over $\bF_q$
for each even $q \geq 4$. \\

In this paper,  three infinite families of distance-optimal repeated-root cyclic codes with minimum distance 3 or 4 are constructed; two other infinite families of repeated-root cyclic codes with minimum distance 3 or 4 are developed; seven infinite families of repeated-root cyclic codes with minimum distance 6 or 8 or 10 are presented; and two infinite families of
repeated-root binary cyclic codes with parameters $[2n, k, d \geq (n-1)/\log_2 n]$, where $n=2^m-1$ and
$k \geq n$, are constructed.
In addition,  27 repeated-root cyclic codes of length up to $254$ over $\bF_q$ for $q \in \{2, 4, 8\}$
with optimal parameters or best parameters known are obtained in this paper.
The results of this paper show that repeated-root cyclic codes could be very attractive and are worth of
further investigation. \\

	\end{abstract}
	
\newpage
\section{Introduction}\label{sec-intro}

\subsection{Cyclic codes and repeated-root cyclic codes}
	
	The Hamming weight $wt_H({\bf a})$ of a vector ${\bf a}=(a_1, \ldots, a_n) \in {\bf F}_q^n$ is the cardinality  of its support $$supp({\bf a})=\{i:a_i \neq 0\}.$$ The Hamming distance $d_H({\bf a}, {\bf b})$ between two vectors ${\bf a}$ and ${\bf b}$ is defined to be the Hamming weight of ${\bf a}-{\bf b}$. For a code ${\bf C} \subseteq {\bf F}_q^n$, its minimum Hamming distance is $$d_H=\min_{{\bf a} \neq {\bf b}} \{d_H({\bf a}, {\bf b}): {\bf a}, {\bf b} \in {\bf C}\}.$$ An $[n, k, d_H]_q$ linear code is a $k$-dimensional linear subspace of $\bF_q^n$
	with minimum Hamming distance $d_H$.
	It is well-known that the Hamming distance of a linear code ${\bf C}$ is the minimum Hamming weight of its non-zero codewords. For the theory of error-correcting codes in the Hamming metric,  the reader is referred to \cite{MScode,Lint,HP}.   A code ${\bf C}$ in ${\bf F}_q^n$ of  minimum distance $d_H$ and cardinality $M$ is called an $(n, M, d_H)_q$ code. For an $[n, k, d_H]_q$ linear code, the Singleton bound asserts that $d_H \leq n-k+1$. When the equality holds, this code is called a maximal distance separable (MDS) code.  Reed-Solomon codes are well-known MDS codes \cite{HP}. \\

We recall the sphere packing bound for $(n,M,d)_q$ codes,
$$M \cdot V_{(q,n)}(\lfloor \frac{d-1}{2}\rfloor) \leq q^n,$$
where $V_{(q,n)}(r)=1+n(q-1)+\displaystyle{n \choose 2}(q-1)^2+\cdots+\displaystyle{n \choose r}(q-1)^r$ is the volume of the ball with radius $r$ in the Hamming metric space $\bF_q^n$ (\cite{HP,Lint,MScode}).
If there is an $(n, M, d)_q$ code and there is no $(n, M, d')_q$ code with $d'>d$, this $(n, M, d)_q$ code is said to be {\it distance-optimal}.   An  $(n, M, d)_q$ code is distance-optimal with respect to the sphere packing bound,
provided that $M \cdot V_{(q,n)}(\lfloor \frac{d}{2}\rfloor) >q^n$.\\

A linear code $\bC \subseteq \bF_q^n$ is called cyclic if $(c_0, c_1, \ldots, c_{n-1}) \in \mathcal{C}$ implies  $(c_{n-1}, c_0, \ldots, c_{n-2}) \in \mathcal{C}$.
The dual code of a cyclic code $\bC$ is also a cyclic code. A codeword $\bc=(c_0, c_1, \ldots, c_{n-1})$ in a cyclic code is identified with the polynomial $\mathbf{c}(x)=c_0+c_1x+\cdots+c_{n-1}x^{n-1}\in \bF_q[x]/(x^n-1)$.
Then every cyclic code $\bC$ is identified with a principal ideal in the ring $\bF_q[x]/(x^n-1)$,
which is generated by a factor $g(x)$ of $x^n-1$ with the smallest degree.  This polynomial $g(x)$ is
called the \emph{generator polynomial} of $\bC$ and $h(x)=(x^n-1)/g(x)$ is called the \emph{check polynomial} of $\bC$.  It is well known that the generator polynomial of the dual code $\bC^\perp$ is
the reciprocal polynomial of $h(x)$.  \\

Let $n$ be a positive integer satisfying $\gcd(n,q)=1$.
Throughout this paper, $i \bmod{n}$ denotes the unique integer $u$ such that
$0 \leq u \leq n-1$ and $i \equiv u \pmod{n}$.
Set $\bZ_n=\bZ/n\bZ=\{0, 1, \ldots, n-1\}$.
For any $a\in \bZ_{n}$, the $q$-cyclotomic coset modulo $n$ of $a$ is defined by
\begin{equation*}
    C_a^{(q, n)}=\{aq^i \bmod {n}: 0\leq i\leq \ell_a-1 \},
\end{equation*}
where $\ell_a$ is the smallest positive integer such that $a q^{\ell_a} \equiv a \pmod {n}$.
The smallest nonnegative integer in $C_a^{(q, n)}$ is called the coset leader.
It is clear that the $q$-cyclotomic cosets correspond to irreducible factors of $x^n-1$ in $\bF_q[x]$.
Therefore, the generator polynomial of a cyclic code of length $n$ over $\bF_q$ is the product of several irreducible factors of $x^n-1$.
The defining set of a cyclic code with generator polynomial ${ g(x)}$ with respect to an
$n$-th primitive root
$\beta$ in an extension field of $\bF_q$ is the the following set
$${\bf T}=\{i: {g}(\beta^i)=0\}.$$
Then the defining set of a cyclic code is the disjoint union of several cyclotomic cosets.
If there are $\delta-1$ consecutive elements in the defining set of a cyclic code,
then the minimum distance of this cyclic code is at least $\delta$.
This is the famous BCH bound for cyclic codes (\cite{BC1,BC2,Hoc,HP,Lint,MScode}).
The reader is referred to \cite{HT,Roos} for the  Hartmann-Tzeng bound and Roos bound for cyclic codes.
Repeated-root cyclic codes over $\bF_q$ are these cyclic codes with  generator polynomial ${g(x)}$ having repeated roots
over an extension field of $\bF_q$.  A cyclic code over $\bF_q$ with generator polynomial having no repeated root
over any extension field of $\bF_q$ is called a simple-root cyclic code.

\subsection{BCH cyclic codes}\label{sec-BCHcode}

Let $\gcd(q,n)=1.$
 Let $m=\ord_n(q)$, which is the smallest positive integer $\ell$ such that $q^\ell \equiv 1 \pmod{n}$.
 Let $\alpha$ be a primitive element of $\bF_{q^m}$. Define $\beta=\alpha^{(q^m-1)/n}$.
 Then $\beta$ is an $n$-th primitive root of unity. Define
 $$
 \m_{\beta^i}(x)=\prod_{j \in C_i^{(q, n)}} (x-\beta^j).
 $$
 It is easily seen that $ \m_{\beta^i}(x)$ is an irreducible polynomial in $\bF_q[x]$ and a factor of
 $x^n-1$.
Let $\delta \geq 2$ be an integer and $b$ be an integer. Define
\begin{eqnarray}\label{eqn-gpBCHcode}
g_{(q,n,\delta,b)}(x)=\lcm\{ \m_{\beta^b}(x), \ldots, \m_{\beta^{b+\delta-2}}(x)\},
\end{eqnarray}
where $\lcm$ denotes the least common multiple of the set of polynomials over $\bF_q$.
Let $\bC_{(q,n,\delta,b)}$ denote the cyclic code over $\bF_q$ with length $n$ and
generator polynomial $g_{(q,n,\delta,b)}(x)$. This code  $\bC_{(q,n,\delta,b)}$
is called a \emph{BCH code} with designed distance $\delta$. It is well known that the
minimum distance of $\bC_{(q,n,\delta,b)}$ is lower bounded by $\delta$.
This follows from the famous BCH bound for cyclic codes \cite{HP,MScode}.
The BCH codes introduced in 1959-1960 (\cite{BC1,BC2,Hoc,Ding4,Ding1}), are one kind of the most important codes in coding theory and practice. Many best known binary codes can be constructed from binary BCH codes.
Some BCH codes and Goppa codes can be considered as subfield subcodes of generalized Reed-Solomon codes, and contain examples of optimal binary codes for many parameters (\cite{MScode,HP,Grassl}).
In this paper, we will use BCH codes to construct repeated-root cyclic codes.
	
\subsection{Motivations and objectives of this paper}

Cyclic codes were introduced by E. Prange \cite{Prange} and there are a lot of references on cyclic codes.
But most of those references are about cyclic codes with simple roots and there are a very small number of
references on cyclic codes with repeated roots \cite{BGG,Massey,HSZ,LY,Lint1,PZL,Sobani,ZU}, although there are more repeated-root cyclic codes than simple-root cyclic codes.  This shows that repeated-root cyclic codes are rarely
studied and much less understood. This is the first motivation of this paper.  \\

Although cyclic codes are a small subclass of linear codes,  many cyclic codes are distance-optimal linear codes
(\cite{Ding0},\cite{Grassl}).  For example, the optimal $[63, 56, 4]_2$ binary code and the optimal $[63, 50, 6]_2$ binary code are BCH codes. However, these distance-optimal cyclic codes are not repeated-root cyclic codes.  In \cite{Lint1} and \cite{HSZ}, two infinite families of distance-optimal repeated-root binary codes with parameters
$[2(2^m-1), 2(2^m-1)-m-2,  4]_2$ were constructed. In \cite{HSZ}, an infinite family of
distance-optimal repeated-root binary codes with parameters
$$
\left[ 2 \frac{4^m-1}{3},   2 \frac{4^m-1}{3}-2m-2, 4     \right]_2
$$
and 39 distance-optimal binary codes with length up to 256 were constructed.
The second motivation is the following questions:
\begin{itemize}
\item What are the distance-optimal repeated-root cyclic codes of length up to $254$ over $\bF_q$ for
$q \in \{4, 8\}$?
\item Is there any infinite family of  distance-optimal repeated-root cyclic codes over $\bF_q$ for each even $q \geq 4$?  If the answer to this question is positive, how can one construct such an infinite family of cyclic codes?
\end{itemize}

It is interesting to construct distance-optimal codes.
Distance-optimal binary codes with minimum distance four and six were constructed in \cite{Ding5,Heng,Wang}.
Only a small number of distance-optimal codes with minimum distance six were reported in the literature (\cite{Heng},  \cite[Theorem 9]{Wang}). It is always a challenging problem to construct distance-optimal codes with larger  minimum distances.
The first objective of this paper is to construct infinite families of distance-optimal repeated-root cyclic codes of small minimum distances with optimal parameters or best parameters known.  \\

It is a long-standing open problem whether there exists an infinite family of asymptotically good cyclic codes \cite{Willems}.  Then it is interesting to construct infinite families of rate $\frac{1}{2}$ explicit cyclic codes such that their minimum distances are as large as possible.  Recently, an infinite family of  $[2^m-1, 2^{m-1}]_2$ binary cyclic codes with square-root-like lower bounds on their minimum distances and dual minimum distances was constructed in \cite{TangDing}.  An infinite family of $[n, \frac{n+1}{2}, d\geq \frac{n}{\log_2 n}]_2$ binary cyclic codes with $n=2^p-1$ and $p$ being an odd prime number, was constructed in \cite{SunLiDing}.
Motivated by these works,  the second objective of this paper is to construct infinite families of
repeated-root binary cyclic codes with parameters $[2n, k, d \geq (n-1)/\log_2 n]$, where $n=2^m-1$ and
$k \geq n$.

\subsection{The organisation of this paper}

The rest of this paper is organised as follows.
Section \ref{sec-vanLintThm} recalls the generalised van Lint theorem, which is the basic tool of this paper.
Section \ref{sec-opt34} constructs several infinite families of repeated-root cyclic codes with minimum distance four or three.
Section \ref{sec-dist68} presents two infinite families of repeated-root cyclic codes with minimum distance six or eight.
Section \ref{sec-dist810} constructs three
infinite families of repeated-root cyclic codes over $\bF_2$ with minimum distance $8$ or at least $10$.
Section \ref{sec-dist6} proposes two infinite families of repeated-root cyclic codes over $\bF_q$ with minimum distance 6 for $q \geq 4$.
Section \ref{sec-largek} constructs two infinite families of repeated-root binary cyclic codes with large dimensions and large minimum distances.
Section \ref{sec-fin} summarises the contributions of this paper and makes some concluding remarks.

\section{The generalized van Lint theorem}\label{sec-vanLintThm}

In this section, we recall the following result in \cite{ChenDing}, which is the basic
 tool for constructing repeated-root cyclic codes in this paper.\\

Let $\bC_1$ and $\bC_2$ be two linear codes with parameters $[n,k_1, d_1]_q$ and $[n,k_2, d_2]_q$, respectively.
The Plotkin sum of $\bC_1$ and $\bC_2$ is denoted by $\Plotkin(\bC_1, \bC_2)$ and defined by
\begin{eqnarray*}
\Plotkin(\bC_1, \bC_2):=\{(\bu|\bu+\bv):  \bu \in \bC_1, \, \bv \in \bC_2 \},
\end{eqnarray*}
where $\bu | \bv$ denotes the concatenation of the two vectors $\bu$ and $\bv$. It is well known that
$\bC$ has parameters $[2n, k_1+k_2, \min\{2d_1, d_2\}]_q$ \cite{MScode}. The Plotkin sum is also called
the  $[{\bf u}|{\bf u}+{\bf v}]$ construction and was introduced in 1960 by Plotkin \cite{Plotkin}. \\

The $[{\bf u}|{\bf u}+{\bf v}]$ construction of binary repeated-root cyclic codes was given by van Lint in his paper \cite{Lint1}.  The following theorem is a  generalization of the original van Lint theorem.

\begin{theorem}[The generalized van Lint theorem \cite{ChenDing}]\label{thm-evanlint}
Let $q$ be a power of $2$ and $n$ be an odd positive integer. Let ${\bf C}_1 \subseteq {\bf F}_q^n$ be a cyclic code
with generator polynomial $g_1(x) \in {\bf F}_q[x]$ and ${\bf C}_2 \subseteq {\bf F}_q^n$ be a cyclic code
generated by the polynomial $g_1(x)g_2(x) \in {\bf F}_q[x]$, where $g_2(x)$ is a divisor of $x^n+1$.
Then the code $\Plotkin(\bC_1, \bC_2)$
is permutation-equivalent to  the repeated-root cyclic code $\bC(g_1, g_2)$ of length $2n$ generated by the polynomial $g_1(x)^2g_2(x)$.  In addition,  the cyclic code $\bC(\bC_1, \bC_2)$ has generator polynomial
$$
\frac{g_1(x)^2 g_2(x)}{\gcd(g_1(x), g_2(x))},
$$
dimension
$$
2n-2\deg(g_1(x))-\deg(g_2(x)) + \deg(\gcd(g_1(x), g_2(x)) )
$$
and minimum distance
$$
d(\bC(\bC_1, \bC_2))=\min\{ 2d(\bC_1), d(\bC_2)\},
$$
here and hereafter $d(\bC)$ denotes the minimum distance of $\bC$.
\end{theorem}

Notice that some of the conclusions in Theorem \ref{thm-evanlint} may be wrong if $q$ is odd.
By definition,  the code $\bC_2$ in Theorem \ref{thm-evanlint} has generator polynomial
$$
\frac{g_1(x) g_2(x)}{\gcd(g_1(x), g_2(x))}
$$
and dimension
$$
n-\deg(g_1(x))-\deg(g_2(x))+\deg(\gcd(g_1(x), g_2(x))).
$$
Therefore, $\bC_2$ is a subcode of $\bC_1$.  \\

In this paper,  two cyclic codes $\bC_1$ and $\bC_2$ of odd length $n$ over $\bF_q$
satisfying $\bC_2 \subseteq \bC_1$ are properly selected or designed, where $q$ is even.
Let $g_1(x)$ and $g_1(x)g_2(x)$ be the generator polynomials of $\bC_1$ and $\bC_2$,
respectively.  We will study the corresponding repeated-root cyclic code $\bC(\bC_1, \bC_2)$.
To make the	 repeated-root cyclic code $\bC(\bC_1, \bC_2)$ have very good or optimal parameters, the two
building blocks $\bC_1$ and $\bC_2$ must be chosen carefully.

\section{Infinite families of repeated-root cyclic codes with minimum distance four or three}\label{sec-opt34}

In this section, we construct several infinite families of repeated-root cyclic codes with minimum distance four or three. To prove the distance optimality of some of them, we need the following lemma.

\begin{lemma}\label{lem-fubound}
\cite{Fuetal}.
Let $\bC$ be an $[n,k,d]$ linear code with even $d$.  Then
$$
\sum_{i=0}^{(d-2)/2} \binom{n-1}{i} (q-1)^i  \leq q^{n-1-k}.
$$
\end{lemma}

The following theorem gives an infinite family of distance-optimal repeated-root binary cyclic codes with minimum distance four and is due to van Lint \cite{Lint1}. We document it here for completeness.

\begin{theorem}\label{thm-20301} \cite{Lint1}
Let $m \geq 4$ be an integer and $n=2^m-1$. Then the repeated-root binary cyclic code $\bC(\bC_{(2, n, 2,  0)}, \bC_{(2, n, 3,0)})$
has parameters $[2(2^m-1), 2(2^m-1)-m-2,  4]_2$ and generator polynomial $(x-1)^2\m_{\beta}(x)$, where
$\beta$ is a primitive element of $\bF_{2^m}$.  In addition, $\bC(\bC_{(2, n, 2,  0)}, \bC_{(2, n, 3,0)})$ is
distance-optimal with respect to the sphere-packing bound.
\end{theorem}

\begin{example}
The first four codes in the family of codes $\bC(\bC_{(2, n, 2,  0)}, \bC_{(2, n, 3,0)})$ have the following parameters:
$$
[30,24,4  ]_2, \ [62,55,4  ]_2, \ [126,118,4  ]_2, \ [254,245,4  ]_2.
$$
All of them are distance-optimal.
\end{example}

The following theorem gives an infinite family of distance-optimal repeated-root cyclic codes
over $\bF_q$ with minimum distance three for even $q \geq 4$.

\begin{theorem}\label{thm-20302}
Let $m \geq 2$ be an integer and $n=q^m-1$, where $q \geq 4$ is even. Then the repeated-root cyclic code
$\bC(\bC_{(q, n, 2,  0)}, \bC_{(q, n, 3,0)})$
has parameters $[2(q^m-1), 2(q^m-1)-m-2,  3]_q$ and generator polynomial $(x-1)^2\m_{\beta}(x)$, where
$\beta$ is a primitive element of $\bF_{q^m}$.  Furthermore,  the code $\bC(\bC_{(q, n, 2,  0)}, \bC_{(q, n, 3,0)})$ is distance-optimal with respect to the bound in Lemma \ref{lem-fubound} and the sphere packing bound.
\end{theorem}

\begin{proof}
{\em
By definition, the BCH code $\bC_{(q, n, 2,  0)}$ has generator polynomial $x-1$. It follows from the BCH bound that $d(\bC_{(q, n, 2,  0)}) \geq 2$. On the other hand,  the codeword $x-1$ in $\bC_{(q, n, 2,  0)}$ has Hamming weight $2$. Consequently, $d(\bC_{(q, n, 2,  0)}) = 2$ and $\bC_{(q, n, 2,  0)}$ has parameters $[n, n-1, 2]_q$. \\

By definition, the BCH code $\bC_{(q, n, 3,  0)}$ has generator polynomial $(x-1)\m_{\beta}(x)$. It follows from the BCH bound that $d(\bC_{(q, n, 3,  0)}) \geq 3$.  We now prove that $\bC_{(q, n, 3,  0)}$ has a codeword of
Hamming weight $3$.  To this end, we consider the following set
\begin{eqnarray*}
\left\{ \frac{\beta^i +1}{\beta^{q^m-2}+1}: 1 \leq i \leq q^m-2    \right\} =J \cup \{1\},
\end{eqnarray*}
where
$$
J=\left\{ \frac{\beta^i +1}{\beta^{q^m-2}+1}: 1 \leq i \leq q^m-3    \right\}.
$$
Notice that $J \cap \{0,1\}=\emptyset$ and $|J|=q^m-3$. There exists an integer $i$ with $1 \leq i \leq q^m-3$
such that
$$
b:= \frac{\beta^i +1}{\beta^{q^m-2}+1} \in \bF_{q} \setminus \{0,1\}.
$$
Put
$$
b_1=\frac{1}{1+b} \in  \bF_{q} \setminus \{0,1\}
$$
and
$$
b_2=\frac{b}{1+b} \in  \bF_{q} \setminus \{0,1\}.
$$
Then $b_1 \neq b_2$. Let
$$
c(x)=1+b_1x^i+b_2x^{q^m-2}.
$$
It is easily verified that $c(1)=c(\beta)=0$. Therefore, $c(x)$ is a codeword in $\bC_{(q, n, 3,  0)}$.
Consequently, $d(\bC_{(q, n, 3,  0)}) = 3$.  It is easily seen that the cyclotomic coset
$$
C_{1}^{(q,n)}=\{1, q,  \ldots, q^{m-1}\}.
$$
Hence $\deg(\m_{\beta}(x))=m$. As a result, $\dim(\bC_{(q, n, 3,  0)})=n-m-1$. Consequently,
$\bC_{(q, n, 3,  0)}$ has parameters $[n, n-m-1, 3]_q$. \\

Since $\m_{\beta}(x)$ is irreducible and has degree $m \geq 2$,  we have that $\gcd(x-1, \m_{\beta}(x))=1$.
By Theorem \ref{thm-evanlint} the code
$\bC(\bC_{(q, n, 2,  0)}, \bC_{(q, n, 3,0)})$ has generator polynomial $(x-1)^2\m_{\beta}(x)$ and parameters $[2(q^m-1), 2(q^m-1)-m-2,  3]_q$.  \\

It follows from the sphere packing bound that any linear code $\bC$ with parameters $[2(q^m-1), 2(q^m-1)-m-2,  d]_q$ must satisfy $d \leq 4$.  It then follows from Lemma \ref{lem-fubound} that there is no linear code
with  parameters $[2(q^m-1), 2(q^m-1)-m-2,  4]_q$ for $q \geq 4$ and $m \geq 2$.  Hence, the code
$\bC(\bC_{(q, n, 2,  0)}, \bC_{(q, n, 3,0)})$ is distance-optimal.
This completes the proof.
}
\end{proof}

\begin{example}
Let $(q, m)=(4,2)$. Then the cyclic code  $\bC(\bC_{(q, n, 2,  0)}, \bC_{(q, n, 3,0)})$ has parameters
$[30, 26, 3]_4$ and is distance-optimal.
Notice the code with the same parameters in \cite{Grassl} is not known to be cyclic.
\end{example}

\begin{example}
Let $(q, m)=(4,3)$. Then the cyclic code  $\bC(\bC_{(q, n, 2,  0)}, \bC_{(q, n, 3,0)})$ has parameters
$[126, 121, 3]_4$ and is distance-optimal.
Notice the code with the same parameters in \cite{Grassl} is not known to be cyclic.
\end{example}

\begin{example}
Let $(q, m)=(8,2)$. Then the cyclic code  $\bC(\bC_{(q, n, 2,  0)}, \bC_{(q, n, 3,0)})$ has parameters
$[126, 122, 3]_8$ and is distance-optimal.
Notice the code with the same parameters in \cite{Grassl} is not known to be cyclic.
\end{example}

The next theorem documents another infinite family of distance-optimal repeated-root cyclic codes of minimum four.

\begin{theorem}\label{thm-2040}
Let $m \geq 2$ be an integer and $n=q^m-1$, where $q \geq 4$ is even. Then the repeated-root cyclic code
$\bC(\bC_{(q, n, 2,  0)}, \bC_{(q, n, 4,0)})$
has parameters $[2(q^m-1), 2(q^m-1)-2m-2,  4]_q$ and generator polynomial
$(x-1)^2\m_{\beta}(x)\m_{\beta^2}(x)$, where
$\beta$ is a primitive element of $\bF_{q^m}$.
In addition, the code $\bC(\bC_{(q, n, 2,  0)}, \bC_{(q, n, 4,0)})$ is distance-optimal with respect to the
sphere packing bound if $qm>8$.
\end{theorem}

\begin{proof}
{\em
It was shown in the proof of Theorem \ref{thm-20302} that $\bC_{(q, n, 2,  0)}$ has parameters $[n, n-1, 2]_q$
and generator polynomial $x-1$. \\

By definition, the BCH code $\bC_{(q, n, 4,  0)}$ has generator polynomial $(x-1)\m_{\beta}(x)\m_{\beta^2}(x)$. It follows from the BCH bound that $d(\bC_{(q, n, 4,  0)}) \geq 4$.  It can be easily verified that
$
|C_1^{(q,n)}|=|C_2^{(q,n)}|=m.
$
Consequently,
$\bC_{(q, n, 4,  0)}$ has parameters $[n, n-2m-1,  d \geq 4]_q$. \\

Since $\m_{\beta^i}(x)$ is irreducible and has degree $m \geq 2$ for $i \in \{1,2\}$,  we have that $\gcd(x-1, \m_{\beta}(x)\m_{\beta^2}(x))=1$.
By Theorem \ref{thm-evanlint} the code
$\bC(\bC_{(q, n, 2,  0)}, \bC_{(q, n, 4,0)})$ has generator polynomial $(x-1)^2\m_{\beta}(x)\m_{\beta^2}(x)$ and parameters $[2(q^m-1), 2(q^m-1)-2m-2,  4]_q$.
It can be verified that $V_{(q,2n)}(2) > q^{2m+2}$ if $qm>8$.  Hence, the code $\bC(\bC_{(q, n, 2,  0)}, \bC_{(q, n, 4,0)})$ is distance-optimal with respect to the
sphere packing bound.
This completes the proof.
}
\end{proof}

\begin{example}
Let $(q, m)=(4,2)$. Then the cyclic code  $\bC(\bC_{(q, n, 2,  0)}, \bC_{(q, n, 4,0)})$ has parameters
$[30, 24, 4]_4$ and is distance-optimal according to  \cite{Grassl}.
Notice that the linear code with the same parameters in  \cite{Grassl} is not known to be cyclic.
\end{example}

\begin{example}
Let $(q, m)=(4,3)$. Then the cyclic code  $\bC(\bC_{(q, n, 2,  0)}, \bC_{(q, n, 4,0)})$ has parameters
$[126, 118, 4]_4$ and is distance-optimal.
  Notice that the linear code with the same parameters in  \cite{Grassl} is not known to be cyclic.
\end{example}

\begin{example}
Let $(q, m)=(8,2)$. Then the cyclic code  $\bC(\bC_{(q, n, 2,  0)}, \bC_{(q, n, 4,0)})$ has parameters
$[126, 120, 4]_8$ and is distance-optimal.
Notice that the linear code with the same parameters in  \cite{Grassl} is not known to be cyclic.
\end{example}

The next theorem presents an infinite family of repeated-root cyclic codes with minimum distance
three or four.

\begin{theorem}\label{thm-20301p}
Let $m \geq 2$ be an integer and $n=(q^m-1)/(q-1)$, where $q \geq 4$ is even. Then the repeated-root binary cyclic code $\bC(\bC_{(q, n, 2,  0)}, \bC_{(q, n, 3,0)})$
has parameters $[2n, 2n-m-2,  3 \leq d \leq 4]_q$ and generator polynomial
$(x-1)^2\m_{\beta}(x)$, where
$\beta$ is a primitive $n$-th root of unity in $\bF_{q^m}$.
\end{theorem}

\begin{proof}
{\em
By definition, the BCH code $\bC_{(q, n, 2,  0)}$ has generator polynomial $x-1$. It follows from the BCH bound that $d(\bC_{(q, n, 2,  0)}) \geq 2$. On the other hand,  the codeword $x-1$ in $\bC_{(q, n, 2,  0)}$ has Hamming weight $2$. Consequently, $d(\bC_{(q, n, 2,  0)}) = 2$ and $\bC_{(q, n, 2,  0)}$ has parameters $[n, n-1, 2]_q$. \\

By definition, the BCH code $\bC_{(q, n, 3,  0)}$ has generator polynomial
$(x-1)\m_{\beta}(x)$. It follows from the BCH bound that $d(\bC_{(q, n, 3,  0)}) \geq 3$.
It can be verified that
$$
\deg(\m_{\beta}(x))=m.
$$
As a result, $\dim(\bC_{(q, n, 3,  0)})=n-m-1$. Consequently,
$\bC_{(q, n, 3,  0)}$ has parameters $[n, n-m-1,  3 \leq d \leq 4]_q$. \\

Since $\gcd(x-1, \m_{\beta}(x) )=1$, by Theorem \ref{thm-evanlint} the code
$\bC(\bC_{(q, n, 2,  0)}, \bC_{(q, n, 3,0)})$ has generator polynomial $(x-1)^2\m_{\beta}(x)$ and parameters $[2n, 2n-m-2,  d\geq 3]_q$.  It follows from the sphere packing bound that
$d(\bC(\bC_{(q, n, 2,  0)}, \bC_{(q, n, 3,0)})) \leq 4$.
This completes the proof.
}
\end{proof}

\begin{example}
The following is a list of codes in the family of the codes  $\bC(\bC_{(q, n, 2,  0)}, \bC_{(q, n, 3,0)})$
in Theorem \ref{thm-20301p}.
\begin{itemize}
\item When $(q, m)=(4 , 2)$, the code $\bC(\bC_{(q, n, 2,  0)}, \bC_{(q, n, 3,0)})$ has parameters $[10, 6, 4 ]_4$.
This code is distance-optimal with respect to the Griesmer bound.  Notice that the linear code with the same parameters in \cite{Grassl} is not known to be cyclic.

\item When $(q, m)=(4 , 3)$, the code $\bC(\bC_{(q, n, 2,  0)}, \bC_{(q, n, 3,0)})$ has parameters $[42,37,3 ]_4$.
This code is distance-optimal according to Theorem \ref{thm-20301p}.  Notice that the linear code with the same parameters in \cite{Grassl} is not known to be cyclic.

\item When $(q, m)=(4 , 4)$, the code $\bC(\bC_{(q, n, 2,  0)}, \bC_{(q, n, 3,0)})$ has parameters $[170,164,3 ]_4$.
This code is distance-optimal according to Theorem \ref{thm-20301p}.  Notice that the linear code with the same parameters in \cite{Grassl} is not known to be cyclic.

\item When $(q, m)=(8 , 2)$, the code $\bC(\bC_{(q, n, 2,  0)}, \bC_{(q, n, 3,0)})$ has parameters $[18,14,4 ]_8$.
This code is distance-optimal according to Theorem \ref{thm-20301p}.  Notice that the linear code with the same parameters in \cite{Grassl} is not known to be cyclic.

\item When $(q, m)=(8 , 3)$, the code $\bC(\bC_{(q, n, 2,  0)}, \bC_{(q, n, 3,0)})$ has parameters $[146,141,3 ]_8$ and is distance-optimal according to Theorem \ref{thm-20301p}.  No best linear code over $\bF_8$ with length $146$ and dimension $141$ is reported in \cite{Grassl}.
\end{itemize}
\end{example}

To the best knowledge of the authors, no infinite family of linear codes with the same length and dimension better than $\bC(\bC_{(q, n, 2,  0)}, \bC_{(q, n, 3,0)})$ is reported in the literature.
The next theorem presents an infinite family of repeated-root cyclic codes with minimum distance four.

\begin{theorem}\label{thm-2040p}
Let $m \geq 2$ be an integer and $n=(q^m-1)/(q-1)$, where $q \geq 4$ is even. Then the repeated-root binary cyclic code $\bC(\bC_{(q, n, 2,  0)}, \bC_{(q, n, 4,0)})$
has parameters $[2n, 2n-2m-2,  4]_q$ and generator polynomial
$(x-1)^2\m_{\beta}(x) \m_{\beta^2}(x)$, where
$\beta$ is a primitive $n$-th root of unity in $\bF_{q^m}$.
\end{theorem}

\begin{proof}
{\em
By definition, the BCH code $\bC_{(q, n, 2,  0)}$ has generator polynomial $x-1$. It follows from the BCH bound that $d(\bC_{(q, n, 2,  0)}) \geq 2$. On the other hand,  the codeword $x-1$ in $\bC_{(q, n, 2,  0)}$ has Hamming weight $2$. Consequently, $d(\bC_{(q, n, 2,  0)}) = 2$ and $\bC_{(q, n, 2,  0)}$ has parameters $[n, n-1, 2]_q$. \\

By definition, the BCH code $\bC_{(q, n, 4,  0)}$ has generator polynomial
$(x-1)\m_{\beta}(x)\m_{\beta^2}(x)$. It follows from the BCH bound that $d(\bC_{(q, n, 4,  0)}) \geq 4$.
It can be verified that
$$
\deg(\m_{\beta}(x))=\deg(\m_{\beta^2}(x))=m.
$$
As a result, $\dim(\bC_{(q, n, 4,  0)})=n-2m-1$. Consequently,
$\bC_{(q, n, 4,  0)}$ has parameters $[n, n-2m-1,  d \geq 4]_q$. \\

Since $\gcd(x-1, \m_{\beta}(x) \m_{\beta^2}(x) )=1$, by Theorem \ref{thm-evanlint} the code
$\bC(\bC_{(q, n, 2,  0)}, \bC_{(q, n, 4,0)})$ has generator polynomial $(x-1)^2\m_{\beta}(x) \m_{\beta^2}(x)$ and parameters $[2n, 2n-2m-2,  4]_q$.
This completes the proof.
}
\end{proof}

\begin{example}
The following is a list of codes in the family of the codes  $\bC(\bC_{(q, n, 2,  0)}, \bC_{(q, n, 4,0)})$
in Theorem \ref{thm-2040p}.
\begin{itemize}
\item When $(q, m)=(4 , 2)$, the code $\bC(\bC_{(q, n, 2,  0)}, \bC_{(q, n, 4,0)})$ has parameters $[10,4,4]_4$.
The distance-optimal linear code with parameters $[10,4,6]_4$ in \cite{Grassl} is not known to be cyclic.

\item When $(q, m)=( 4, 3)$, the code $\bC(\bC_{(q, n, 2,  0)}, \bC_{(q, n, 4,0)})$ has parameters $[42,34,4 ]_4$.
The best linear code with parameters $[42,34,5]_4$ in \cite{Grassl} is not known to be cyclic.

\item When $(q, m)=(4,4 )$, the code $\bC(\bC_{(q, n, 2,  0)}, \bC_{(q, n, 4,0)})$ has parameters $[170,160,4 ]_4$.
The best linear code with parameters $[170,160,5]_4$ in \cite{Grassl} is not known to be cyclic.

\item When $(q, m)=(8 ,2 )$, the code $\bC(\bC_{(q, n, 2,  0)}, \bC_{(q, n, 4,0)})$ has parameters $[ 18,12,4]_8$.
The distance-optimal linear code with parameters $[18,12,6]_8$ in \cite{Grassl} is not known to be cyclic.

\item When $(q, m)=(8 , 3)$, the code $\bC(\bC_{(q, n, 2,  0)}, \bC_{(q, n, 4,0)})$ has parameters $[146,138,4 ]_8$.  No best linear code over $\bF_8$ with length $146$ and dimension $138$ is reported in \cite{Grassl}.

\end{itemize}
\end{example}

To the best knowledge of the authors, no infinite family of linear codes with the same length and dimension better than $\bC(\bC_{(q, n, 2,  0)}, \bC_{(q, n, 4,0)})$ is reported in the literature.  \\

Combining Theorems \ref{thm-20301},  \ref{thm-20302} and \ref{thm-2040}, we deduce the following existence result.

\begin{theorem}
For each even $q \geq 2$, there is an infinite family of distance-optimal repeated-root cyclic codes of minimum
distance four.

For each even $q \geq 4$, there is an infinite family of distance-optimal repeated-root cyclic codes of minimum
distance three.
\end{theorem}

The following theorem describes an infinite family of distance-optimal repeated-root cyclic codes over $\bF_4$.

\begin{theorem}\label{thm-2040h}
Let $m \geq 2$ be an integer and $n=2^{2m-1}-1$. Then the repeated-root cyclic code
$\bC(\bC_{(4, n, 2,  0)}, \bC_{(4, n, 4,0)})$
has parameters $[2n, 2n-2m-1,  4]_4$.  Furthermore,  the code $\bC(\bC_{(4, n, 2,  0)}, \bC_{(4, n, 4,0)})$
is distance-optimal with respect to the sphere packing bound for $m \geq 3$.
\end{theorem}

\begin{proof}
{\em
By definition, the BCH code $\bC_{(4, n, 2,  0)}$ has generator polynomial $x-1$ and thus dimension $n-1$. It follows from the BCH bound that $d(\bC_{(4, n, 2,  0)}) \geq 2$. It then follows from the Singleton bound that
$d(\bC_{(4, n, 2,  0)}) = 2$.  Hence,  $\bC_{(4, n, 2,  0)}$ has parameters $[n, n-1, 2]_4$.

Since $4^{2m}-2=2(2^{2m-1}-1)$,  $2$ is in the $4$-cyclotomic coset $C_1^{(4,n)}$.
By definition, the BCH code $\bC_{(4, n, 4,  0)}$ has generator polynomial $(x-1)\m_{\beta}(x)$.
It can be verified that
$
|C_1^{(4,n)}|=2m-1.
$
Consequently,
$$
\dim(\bC_{(4, n, 4,  0)})=n-2m.
$$
It follows from the BCH bound that
$
d(\bC_{(4, n, 4,  0)}) \geq 4.
$
Note that $\gcd(x-1, m_{\beta}(x))=1$. The desired conclusions on the parameters of the code
$\bC(\bC_{(4, n, 2,  0)}, \bC_{(4, n, 4,0)})$ then follow from Theorem \ref{thm-evanlint}.
It is straightforward to verify that the code $\bC(\bC_{(4, n, 2,  0)}, \bC_{(4, n, 4,0)})$
is distance-optimal with respect to the sphere packing bound for $m \geq 3$.
This completes the proof.
}
\end{proof}

\begin{example}
The first three codes in the family of codes $\bC(\bC_{(4,n,2,0)}, \bC_{(4,n,4,0)})$ are listed below.
\begin{itemize}
\item When $m= 2$, the code $\bC(\bC_{(4,n,2,0)}, \bC_{(4,n,4,0)})$ has parameters $[14,9,4  ]_4$ and is optimal according to \cite{Grassl}.
Notice that the best linear code known with parameters $[14,9,4]_4$ in \cite{Grassl} is not known to be cyclic.
\item When $m= 3$, the code $\bC(\bC_{(4,n,2,0)}, \bC_{(4,n,4,0)})$ has parameters $[62,55,4  ]_4$ and is
 optimal.
Notice that the best linear code known with parameters $[62,55,4]_4$ in \cite{Grassl} is not known to be cyclic.
\item When $m= 4$, the code $\bC(\bC_{(4,n,2,0)}, \bC_{(4,n,4,0)})$ has parameters $[254,245,4  ]_4$ and is
 optimal.
Notice that the best linear code known with parameters $[254,245,4]_4$ in \cite{Grassl} is not known to be cyclic.
\end{itemize}
\end{example}

\section{Infinite families of repeated-root cyclic codes over $\bF_2$ with best parameters known}\label{sec-dist68}

The following theorem documents an infinite family of repeated-root binary cyclic codes with
minimum distance $6$.

\begin{theorem}\label{thm-3160}
Let $m \geq 3$ be an integer and $n=2^m-1$. Then the repeated-root binary cyclic code
$\bC(\bC_{(2, n, 3,  1)}, \bC_{(2, n, 6,0)})$
has parameters $[2(2^m-1), 2(2^m-1)-3m-1,  6]_2$ and generator polynomial
$(x-1)\m_{\beta}(x)^2 \m_{\beta^3}(x)$, where
$\beta$ is a primitive element of $\bF_{2^m}$.
\end{theorem}

\begin{proof}
{\em
By definition, the BCH code $\bC_{(2, n, 3,  1)}$ has generator polynomial $\m_{\beta}(x)$ and parameters
$[n, n-m, 3]_2$, as it is the binary Hamming cyclic code. \\

By definition, the BCH code $\bC_{(2, n, 6,  0)}$ has generator polynomial $(x-1)\m_{\beta}(x)\m_{\beta^3}(x)$. It follows from the BCH bound that $d(\bC_{(2, n, 6,  0)}) \geq 6$.  It can be easily verified that
$
|C_1^{(2,n)}|=|C_3^{(2,n)}|=m.
$
Consequently,
$\bC_{(2, n, 6,  0)}$ has parameters $[n, n-2m-1,  d \geq 6]_2$. \\

Since $\m_{\beta^i}(x)$ is irreducible and has degree $m \geq 3$ for $i \in \{1,3\}$,  we deduce that
$\gcd(\m_{\beta}(x),  (x-1)\m_{\beta^3}(x))=1$.
By Theorem \ref{thm-evanlint} the code
$\bC(\bC_{(2, n, 3,  1)}, \bC_{(2, n, 6,0)})$ has generator polynomial $(x-1)\m_{\beta}(x)^2\m_{\beta^3}(x)$ and parameters $[2(2^m-1), 2(2^m-1)-3m-1,  6]_2$.
This completes the proof.
}
\end{proof}

\begin{example}
The following is a list of the first five codes in the family of binary cyclic codes $\bC(\bC_{(2, n, 3,  1)}, \bC_{(2, n, 6,0)})$.
\begin{itemize}
\item When $m= 3$,  the cyclic code $\bC(\bC_{(2, n, 3,  1)}, \bC_{(2, n, 6,0)})$ has parameters $[14,4,6]_2$.
Notice that the best linear code known with parameters $[14,4,7]_2$ in \cite{Grassl} is not known to be cyclic.
\item When $m= 4$,  the cyclic code $\bC(\bC_{(2, n, 3,  1)}, \bC_{(2, n, 6,0)})$ has parameters $[30,17,6]_2$.
Notice that the best  linear code known with parameters $[30,17,6]_2$ in \cite{Grassl} is not known to be cyclic.
\item When $m= 5$,  the cyclic code $\bC(\bC_{(2, n, 3,  1)}, \bC_{(2, n, 6,0)})$ has parameters $[62,46,6]_2$.
Notice that the best linear code known with parameters $[62,46,6]_2$ in \cite{Grassl} is not known to be cyclic.
\item When $m= 6$,  the cyclic code $\bC(\bC_{(2, n, 3,  1)}, \bC_{(2, n, 6,0)})$ has parameters $[126,107,6]_2$.
Notice that the best linear code known with parameters $[126,107,6]_2$ in \cite{Grassl} is not known to be cyclic.
\item When $m= 7$,  the cyclic code $\bC(\bC_{(2, n, 3,  1)}, \bC_{(2, n, 6,0)})$ has parameters $[254,232,6]_2$.
Notice that the best linear code known with parameters $[254,232,6]_2$ in \cite{Grassl} is not known to be cyclic.
\end{itemize}
In all the five cases, the upper bound on the minimum distance of the code $\bC(\bC_{(2, n, 3,  1)}, \bC_{(2, n, 6,0)})$ is $7$ according to  \cite{Grassl}.  Hence,  in these five cases  the code $\bC(\bC_{(2, n, 3,  1)}, \bC_{(2, n, 6,0)})$ is distance-almost-optimal.
\end{example}

To the best knowledge of the authors, no infinite family of binary linear codes with the same length and dimension better than $\bC(\bC_{(2, n, 3,  1)}, \bC_{(2, n, 6,0)})$ is reported in the literature.
The following theorem documents an infinite family of repeated-root binary cyclic codes with minimum distance $8$.

\begin{theorem}\label{thm-3080}
Let $m \geq 5$ be an integer and $n=2^m-1$. Then the repeated-root binary cyclic code
$\bC(\bC_{(2, n, 3,  0)}, \bC_{(2, n, 8,0)})$
has parameters $[2(2^m-1), 2(2^m-1)-4m-2,  8]_2$ and generator polynomial
$(x-1)^2 \m_{\beta}(x)^2 \m_{\beta^3}(x)  \m_{\beta^5}(x)$, where
$\beta$ is a primitive element of $\bF_{2^m}$.
\end{theorem}

\begin{proof}
{\em
By definition, the BCH code $\bC_{(2, n, 3,  0)}$ has generator polynomial $(x-1)\m_{\beta}(x)$ and parameters
$[n, n-m-1, 4]_2$, as it is the even-weight subcode of the binary Hamming cyclic code. \\

By definition, the BCH code $\bC_{(2, n, 8,  0)}$ has generator polynomial
$(x-1)\m_{\beta}(x)\m_{\beta^3}(x)   \m_{\beta^5}(x)$. It follows from the BCH bound that $d(\bC_{(2, n, 8,  0)}) \geq 8$.  It can be easily verified that
$
|C_1^{(2,n)}|=|C_3^{(2,n)}|= |C_5^{(2,n)}|=m
$
for $m \geq 5$.
Consequently,
$\bC_{(2, n, 8,  0)}$ has parameters $[n, n-3m-1,  d \geq 8]_2$ for $m \geq 5$. \\

Since $\m_{\beta^i}(x)$ is irreducible and has degree $m \geq 5$ for $i \in \{1,3,5\}$,  we deduce that
$\gcd(\m_{\beta}(x)(x-1),  \m_{\beta^3}(x)  \m_{\beta^5}(x) )=1$.
By Theorem \ref{thm-evanlint} the code
$\bC(\bC_{(2, n, 3,  0)}, \bC_{(2, n, 8,0)})$ has generator polynomial
$(x-1)^2\m_{\beta}(x)^2\m_{\beta^3}(x)  \m_{\beta^5}(x)$ and parameters
$[2(2^m-1), 2(2^m-1)-4m-2,  8]_2$.
This completes the proof.
}
\end{proof}

\begin{example}
The following is a list of the first four codes in the family of binary cyclic codes $\bC(\bC_{(2, n, 3,  0)}, \bC_{(2, n, 8,0)})$.
\begin{itemize}
\item When $m= 4$,  the cyclic code $\bC(\bC_{(2, n, 3,  0)}, \bC_{(2, n, 8,0)})$ has parameters $[30,14,8]_2$
and is distance-optimal according to \cite{Grassl}.
Notice that the best  linear code known with parameters $[30,14,8]_2$ in \cite{Grassl} is not known to be cyclic.  Note that the the parameters of the code in the case $m=4$ are not given by Theorem \ref{thm-3080}.
\item When $m= 5$,  the cyclic code $\bC(\bC_{(2, n, 3,  0)}, \bC_{(2, n, 8,0)})$ has parameters $[62,40,8]_2$.
Notice that the best linear code known with parameters $[62,40,8]_2$ in \cite{Grassl} is not known to be cyclic.
\item When $m= 6$,  the cyclic code $\bC(\bC_{(2, n, 3,  0)}, \bC_{(2, n, 8,0)})$ has parameters $[126,100,8]_2$.
Notice that the best linear code known with parameters $[126,100,8]_2$ in \cite{Grassl} is not known to be cyclic.
\item When $m= 7$,  the cyclic code $\bC(\bC_{(2, n, 3,  0)}, \bC_{(2, n, 8,0)})$ has parameters $[254,224,8]_2$.
Notice that the best linear code known with parameters $[254,224,8]_2$ in \cite{Grassl} is not known to be cyclic.
\end{itemize}
In all the four cases,  the cyclic code $\bC(\bC_{(2, n, 3,  0)}, \bC_{(2, n, 8,0)})$ has the same parameters as the best linear code in  \cite{Grassl}.
\end{example}

To the best knowledge of the authors, no infinite family of binary linear codes with the same length and dimension better than $\bC(\bC_{(2, n, 3,  0)}, \bC_{(2, n, 8,0)})$ is reported in the literature.

\section{Infinite families of repeated-root cyclic codes over $\bF_2$ with minimum distance $8$ or at least $10$}\label{sec-dist810}

In this section, we construct several families binary repeated-root cyclic codes with minimum distances $8$ or at least $10$.

\begin{theorem}\label{thm-281}
Let $m \geq 8$ be an even integer and $n=\frac{2^m-1}{3}$. Then the repeated-root binary cyclic code
$\bC(\bC_{(2, n, 4,  0)}, \bC_{(2, n, 8,0)})$
has parameters $[\frac{2(2^m-1)}{3}, \frac{2(2^m-1)}{3}-4m-2, 8]_2$ and generator polynomial
$(x-1)^2\m_{\beta}(x)^2 \m_{\beta^3}(x)\m_{\beta^5}(x)$, where
$\beta$ is an $n$-th root of unity in $\bF_{2^m}$.
\end{theorem}

\begin{proof}
{\em
    For $m \geq 8$ being even, $3\mid 2^m-1$, and it can be easily verified that $|C_1^{(2,n)}|=|C_3^{(2,n)}|= |C_5^{(2,n)}|=m$ and $C_1^{(2,n)}$, $C_3^{(2,n)}$ and $C_5^{(2,n)}$ are pairwise disjoint.
    By definition, the BCH code $\bC_{(2, n, 4, 0)}$ has generator polynomial $(x-1)\m_{\beta}(x)$ and thus
    dimension $n-m-1$.  It follows from the BCH bound that $d(\bC_{(2, n, 4, 0)})\geq 4$.  It follows from
    the sphere packing bound that $d(\bC_{(2, n, 4, 0)})\leq 4$. Consequently, $d(\bC_{(2, n, 4, 0)}) = 4$ and
     $\bC_{(2, n, 4, 0)}$ has parameters $[n, n-m-1, 4]_2$.
     It can be verified that the BCH code $\bC_{(2, n, 8,  0)}$ has generator polynomial $(x-1)\m_{\beta}(x)\m_{\beta^3}(x) \m_{\beta^5}(x)$.
    It follows from the BCH bound that $d(\bC_{(2, n, 8,  0)}) \geq 8$.
    Then the BCH code $\bC_{(2, n, 8,  0)}$ has parameters $[n, n-3m-1, d \geq 8]_2$. \\

    By Theorem \ref{thm-evanlint}, the code $\bC(\bC_{(2, n, 4, 0)}, \bC_{(2, n, 8,0)})$ has generator polynomial $(x-1)^2\m_{\beta}(x)^2$ $\m_{\beta^3}(x) \m_{\beta^5}(x)$ and parameters $[\frac{2(2^m-1)}{3}, \frac{2(2^m-1)}{3}-4m-2, 8]_2$.  The proof is completed.

    }
\end{proof}

\begin{theorem}\label{thm-282}
Let $m \geq 5$ be an odd integer and let $n=3(2^m-1)$. Then the repeated-root binary cyclic code $\bC(\bC_{(2, n, 4,  0)}, \bC_{(2, n, 8,0)})$ has parameters $[6(2^m-1), 6(2^m-1)-7m-2,  8]_2$ and generator polynomial $(x-1)^2\m_{\beta}(x)^2 \m_{\beta^3}(x)\m_{\beta^5}(x)$, where $\beta$
 is an $n$-th root of unity in $\bF_{2^{2m}}.$
\end{theorem}

\begin{proof}
{\em
Recall that $n = 3(2^{m}-1)$.  It can be easily verified that $\ord_n(2)=2m$.
    For $m\geq 5$ being odd,  we have that $|C_1^{(2, n)}|=\ord_n(2) = 2m$.
    Let $|C_3^{(2, n)}|=m_3$. Then by definition $3\cdot2^{m_3} \equiv 3 \pmod n$.
    We deduce that $m_3 = m$ since $3(2^{m}-1) \mid 3(2^{m_3}-1)$.
    Similarly, let $|C_5^{(2, n)}|=m_5 $, we have $3(2^{m}-1) \mid 5(2^{m_5}-1)$.
    Then $3 \mid 5(2^{m_5}-1)$ if and only if $m_5$ is even, since $\gcd(2^2 -1, 2^{m_5}-1) =2^{\gcd(2, m_5)}-1$.
    On the other hand, because $m$ is odd, $\gcd(2^2 +1, 2^{m}-1) = 1$. Thus $(2^{m}-1) \mid 5(2^{m_5}-1)$ if and only if $m\mid m_5$.
    It implies that $|C_5^{(2, n)}| = 2m$.
    It is easily seen that
    $$
    \gcd(n, 2^j-5)=1
    $$
    for each $j$ with $0 \leq j \leq 2m-1$. Consequently,
    $1$ and $5$ are not in the same cyclotomic coset. \\

 By definition, the BCH code $\bC_{(2, n, 4, 0)}$ has generator polynomial $(x-1)\m_{\beta}(x)$ and thus
    dimension $n-2m-1$.  It follows from the BCH bound that $d(\bC_{(2, n, 4, 0)})\geq 4$.  It follows from
    the sphere packing bound that $d(\bC_{(2, n, 4, 0)})\leq 4$. Consequently, $d(\bC_{(2, n, 4, 0)}) = 4$ and
     $\bC_{(2, n, 4, 0)}$ has parameters $[n, n-2m-1, 4]_2$.
    It can be verified that the BCH code $\bC_{(2, n, 8,  0)}$ has generator polynomial $(x-1)\m_{\beta}(x)\m_{\beta^3}(x) \m_{\beta^5}(x)$.  Consequently,
    $$
    \dim(\bC_{(2, n, 8,  0)})=n-(1+|C_1^{(2, n)}|+|C_3^{(2, n)}|+|C_5^{(2, n)}|)=n-5m-1.
    $$
    It follows from the BCH bound that $d(\bC_{(2, n, 8,  0)}) \ge 8$.  Hence, $\bC_{(2, n, 8,  0)}$ has
    parameters $[n, n-5m-1, d \geq 8]_2$. \\

    According to Theorem \ref{thm-evanlint}, the code $\bC(\bC_{(2, n, 4, 0)}, \bC_{(2, n, 8,0)})$ has generator polynomial $(x-1)^2\m_{\beta}(x)^2$ $\m_{\beta^3}(x) \m_{\beta^5}(x)$ and parameters $[6(2^m-1), 6(2^m-1)-7m-2,  8]_2$. The proof is completed.

}
\end{proof}

\begin{theorem}\label{thm-210}
Let $m \geq 5$ be an odd integer and $n=3(2^m-1)$. Then the repeated-root binary cyclic code $\bC(\bC_{(2, n, 5, 1)}, \bC_{(2, n, 10, n-4)})$ has parameters $[6(2^m-1), 6(2^m-1)-9m-1,  d \geq 10]_2$ and generator polynomial $(x-1)\m_{\beta}(x)^2 \m_{\beta^3}(x)^2 \m_{\beta^{-1}}(x)\m_{\beta^{-3}}(x)$, where $\beta$
 is an $n$-th root of unity in $\bF_{2^{2m}}$.
\end{theorem}

\begin{proof}
{\em
    Let $m\geq 5$ be an odd integer and $n = 3(2^{m}-1)$. According to the proof of Theorem \ref{thm-282}, we have that $|C_1^{(2, n)}| = \mathrm{ord}_n(2) =2m$ and $|C_3^{(2, n)}| = m$.
It can be verified that $|C_{n-1}^{(2, n)}| = 2m$ and $|C_{n-3}^{(2, n)}| = m$.
Furthermore,  we claim that $C_1^{(2, n)} \neq C_{n-1}^{(2, n)}$.
    If $C_1^{(2, n)} = C_{n-1}^{(2, n)}$, then both 1 and $-1\bmod n$ are in the set $C = \{2^i \bmod n, n-2^j \bmod n ~|~ 0 \leq i, j \leq m \}$.
    Since $n-2^{m}=2^{m+1} - 3 >2^{m}$, then $|C| = 2m+2 > \mathrm{ord}_n(2) = 2m$, which  is a contradiction.
    The proof of $C_3^{(2, n)} \neq C_{n-3}^{(2, n)}$ is similar and skipped.\\

    By definition and the BCH bound,  the BCH code $\bC_{(2, n, 5, 1)}$ has generator polynomial $\m_{\beta}(x)\m_{\beta^3}(x)$ and parameters $[n, n-3m, d \geq 5]_2$, and
    the BCH code $\bC_{(2, n, 10, n-4)}$ has generator polynomial $(x-1)\m_{\beta}(x)\m_{\beta^3}(x)$ $\m_{\beta^{-1}}(x)\m_{\beta^{-3}}(x)$ and parameters $[n, n-6m-1, d \geq 10]_2$. \\

     By Theorem \ref{thm-evanlint},  the cyclic code $\bC(\bC_{(2, n, 5, 1)}, \bC_{(2, n, 10, n-4)})$ has generator polynomial $(x-1)\m_{\beta}(x)^2 \m_{\beta^3}(x)^2$ $\m_{\beta^{-1}}(x)\m_{\beta^{-3}}(x)$ and parameters $[6(2^m-1), 6(2^m-1)-9m-1, d \geq 10]_2$. The proof is finished.

}
\end{proof}

\begin{example}
The following are some binary repeated-root cyclic codes obtained in Theorems \ref{thm-281}, \ref{thm-282} and \ref{thm-210}.
\begin{itemize}
\item When $m= 8$,  the binary repeated-root cyclic code in Theorem \ref{thm-281} has parameters $[170,136,8]_2$.
Notice that the best known minimum distance of a $[170, 136]_2$ linear code in \cite{Grassl} is $10$,
which is not known to be cyclic.

\item When $m= 5$,  the binary repeated-root cyclic code in Theorem \ref{thm-282} has parameters $[186,149,8]_2$.
Notice that the best known minimum distance of a $[186, 149]_2$ linear code in \cite{Grassl} is $10$,
which is not known to be cyclic.

\item When $m= 5$,  the binary repeated-root cyclic code in Theorem \ref{thm-210} has parameters $[186,140,10]_2$.
Notice that the best known minimum distance of a $[186, 140]_2$ linear code in \cite{Grassl} is $12$,
which is not known to be cyclic.
\end{itemize}
\end{example}

\section{Two infinite families of repeated-root cyclic codes over $\bF_q$ with minimum distance 6 for $q \geq 4$}\label{sec-dist6}

In this section, we construct two infinite families of repeated-root cyclic codes over $\bF_q$ with minimum distance 6 for $q \geq 4$.

\begin{theorem}\label{thm-2060}
Let $m \geq 2$ be an integer and $n=q^m-1$, where $q \geq 4$ is even. Then the repeated-root cyclic code
$\bC(\bC_{(q, n, 3,  0)}, \bC_{(q, n, 6,0)})$
has parameters $[2(q^m-1), k,  6]_q$,
where
\begin{itemize}
\item $k=2(q^m-1)-4m-2$  if $q=4$,  and
\item  $k=2(q^m-1)-5m-2$  if $q>4$,.
\end{itemize}
\end{theorem}

\begin{proof}
{\em
By definition, the BCH code $\bC_{(q, n, 3,  0)}$ has generator polynomial $(x-1)\m_{\beta}(x)$.
It was shown in the proof of Theorem \ref{thm-20302} that
$\bC_{(q, n, 3,  0)}$ has parameters $[n, n-m-1, 3]_q$. \\

It follows from the BCH bound that $d(\bC_{(q, n, 6,0)}) \geq 6.$ We now determine the generator polynomial
and the dimension of the code $\bC_{(q, n, 6,0)}$.  When $q=4$, the generator polynomial of the code
$\bC_{(4, n, 6,0)}$ is $(x-1)\m_{\beta}(x) \m_{\beta^2}(x) \m_{\beta^3}(x)$. It can be verified that
$$
|C_1^{(4,n)}|=|C_2^{(4,n)}|=|C_3^{(4,n)}|=m.
$$
Therefore,  $\dim(\bC_{(4, n, 6,0)})=n-3m-1$.
When $q>4$, the generator polynomial of the code
$\bC_{(q, n, 6,0)}$ is $(x-1)\m_{\beta}(x) \m_{\beta^2}(x) \m_{\beta^3}(x)  \m_{\beta^4}(x)$. It can be verified that
$$
|C_1^{(q,n)}|=|C_2^{(q,n)}|=|C_3^{(q,n)}|=|C_4^{(q,n)}|=m.
$$
Therefore,  $\dim(\bC_{(q, n, 6,0)})=n-4m-1$.

Since $\gcd(n, q)=1$, $x^n-1$ has no repeated roots in $\bF_{q^m}$.
We have then
$$\gcd((x-1)\m_{\beta}(x),  \m_{\beta^2}(x) \m_{\beta^3}(x))=1$$ if $q=4$
and
$$\gcd((x-1)\m_{\beta}(x),  \m_{\beta^2}(x) \m_{\beta^3}(x) \m_{\beta^4}(x))=1$$
if $q>4$ .
The desired conclusions then follow from Theorem \ref{thm-evanlint}.
This completes the proof.
}
\end{proof}

\begin{example}
The following is a list of  codes in the family of cyclic codes $\bC(\bC_{(q, n, 3,  0)}, \bC_{(q, n, 6,0)})$ for $q \geq 4$.
\begin{itemize}
\item When $(q,m)=(4,2)$,  the cyclic code $\bC(\bC_{(q, n, 3,  0)}, \bC_{(q, n, 6,0)})$ had parameters
$[30,20,6 ]_4$.
Notice that the best linear code known with parameters $[30,20,6]_4$ in \cite{Grassl} is not known to be cyclic.

\item When $(q,m)=(4,3)$,  the cyclic code $\bC(\bC_{(q, n, 3,  0)}, \bC_{(q, n, 6,0)})$ had parameters
$[126,112,6 ]_4$.
Notice that the best linear code known with parameters $[126,112,6]_4$ in \cite{Grassl} is not known to be cyclic.

\item When $(q,m)=(8,2)$,  the cyclic code $\bC(\bC_{(q, n, 3,  0)}, \bC_{(q, n, 6,0)})$ had parameters
$[126,114,6 ]_8$.
Notice that the best linear code known with parameters $[126,114,7]_8$ in \cite{Grassl} is not known to be cyclic.
\end{itemize}

\end{example}

To the best knowledge of the authors, no infinite family of linear codes with the same length and dimension better than $\bC(\bC_{(q, n, 3,  0)}, \bC_{(q, n, 6,0)})$ for $q\geq 4$ is reported in the literature.
The following theorem describes an infinite family of repeated-root cyclic codes over $\bF_4$.

\begin{theorem}\label{thm-3160h}
Let $m \geq 2$ be an integer and $n=2^{2m-1}-1$. Then the repeated-root cyclic code
$\bC(\bC_{(4,n,3,1)}, \bC_{(4,n,6,0)})$ has parameters $[2n, 2n-6m+2, 6]_4$.
\end{theorem}

\begin{proof}
{\em
By definition, the BCH code $\bC_{(4, n, 3,  1)}$ has generator polynomial $\m_{\beta}(x)$.
It can be verified that
$$
|C_1^{(4,n)}|=2m-1.
$$
Consequently,
$$
\dim(\bC_{(4, n, 3,  1)})=n-2m+1.
$$
It follows from the BCH bound that $d(\bC_{(4, n, 3,  1)}) \geq 3$.
It then follows from the sphere packing bound  that
$d(\bC_{(4, n, 3, 1)})  \leq 4$.
It  follows from the bound of  Lemma \ref{lem-fubound} that
$d(\bC_{(4, n, 3, 1)})  \neq 4$.
Consequently,
$d(\bC_{(4, n, 3, 1)})  = 3$ and the code  $\bC_{(4, n, 3,  1)}$ has parameters $[n, n-2m+1, 3]_4$.

Since $4^{2m}-2=2(2^{2m-1}-1)$,  $2$ is in the $4$-cyclotomic coset $C_1^{(4,n)}$.
By definition, the BCH code $\bC_{(4, n, 6,  0)}$ has generator polynomial
$(x-1)\m_{\beta}(x)\m_{\beta^3}(x)$.
It can be verified that
$$
|C_3^{(4,n)}|=2m-1.
$$
Consequently,
$$
\dim(\bC_{(4, n, 6,  0)})=n-4m+1
$$
It follows from the BCH bound that
$
d(\bC_{(4, n, 6,  0)}) \geq 6.
$
Note that $\gcd(\m_{\beta}(x), (x-1)m_{\beta^3}(x))=1$. The desired conclusions then follow from Theorem \ref{thm-evanlint}.
}
\end{proof}

\begin{example}
The first three codes in the family of codes $\bC(\bC_{(4,n,3,1)}, \bC_{(4,n,6,0)})$ are listed below.
\begin{itemize}
\item When $m= 2$, the code $\bC(\bC_{(4,n,3,1)}, \bC_{(4,n,6,0)})$ has parameters $[14,4,6  ]_4$.
Notice that the best linear code known with parameters $[14,4,9]_4$ in \cite{Grassl} is not known to be cyclic.
\item When $m= 3$, the code $\bC(\bC_{(4,n,3,1)}, \bC_{(4,n,6,0)})$ has parameters $[62,46,6  ]_4$.
Notice that the best linear code known with parameters $[62,46,8]_4$ in \cite{Grassl} is not known to be cyclic.
\item When $m= 4$, the code $\bC(\bC_{(4,n,3,1)}, \bC_{(4,n,6,0)})$ has parameters $[254,232,6  ]_4$.
Notice that the best linear code known with parameters $[254,232,8]_4$ in \cite{Grassl} is not known to be cyclic.
\end{itemize}
\end{example}

To the best knowledge of the authors, no infinite family of linear codes with the same length and dimension better than $\bC(\bC_{(4,n,3,1)}, \bC_{(4,n,6,0)})$ is reported in the literature.

\section{Infinite families of repeated-root binary cyclic codes with large dimensions and large minimum distances}\label{sec-largek}

The following theorem describes an infinite family of repeated-root binary cyclic codes with parameters
 $[2n, n, d\geq \frac{n}{\log_2 n}]_2$.

\begin{theorem}\label{thm-new111}
Let $m \geq 3$ be a prime and $n=2^m-1$. Put $h=\frac{2^{m-1}-1}{m}$. Let $\delta_h$ denote the
$h$-th largest nonzero $2$-cyclotomic coset leader modulo $n$.
 Then the repeated-root cyclic code $\bC(\bC_{(2,n, \delta_h+2,1)},  \bC_{(2,n, \delta_h+2,0)})$
has parameters $[2n, n, d]_2$, where
\begin{equation}\label{eqn-sldbound}
d \geq \delta_h+ 3 \geq \frac{n-1+2m}{m}.
\end{equation}
\end{theorem}

\begin{proof}
{\em
By definition, each nonzero $2$-cyclotomic coset leader modulo $n$ must be a positive odd integer. Consequently,
$$
\delta_h \geq 2h-1=\frac{n-1}{m} - 1.
$$
It follows from the BCH bound that
$$
d(\bC_{(2,n, \delta_h+2,1)}) \geq \delta_h +2.
$$
Since $m$ is a prime, each nonzero $2$-cyclotomic coset must have cardinality $m$. Consequently,
$$
\dim(\bC_{(2,n, \delta_h+2,1)})=2^{m-1}.
$$
Let $g(x)$ denote the generator polynomial of $\bC_{(2,n, \delta_h+2,1)}$. By definition,
$\bC_{(2,n, \delta_h+2,0)}$ has generator polynomial $(x-1)g(x)$, and is thus the even-weight
subcode of $\bC_{(2,n, \delta_h+2,1)}$.  As a result,
$$
\dim(\bC_{(2,n, \delta_h+2,0)})=2^{m-1}-1.
$$
It is known that the extended code of $\bC_{(2,n, \delta_h+2,1)}$ is affine-invariant. Hence, the minimum distance
of $\bC_{(2,n, \delta_h+2,1)}$ must be odd. Consequently,
$$
d(\bC_{(2,n, \delta_h+2,0)}) \geq d(\bC_{(2,n, \delta_h+2,1)}) + 1 \geq \delta_h +3.
$$
It then follows from Theorem \ref{thm-evanlint} that
\begin{eqnarray*}
&& d(\bC(\bC_{(2,n, \delta_h+2,1)},  \bC_{(2,n, \delta_h+2,0)})) \\
&& \ \ = \min\{ 2d(\bC_{(2,n, \delta_h+2,1)}),  d(\bC_{(2,n, \delta_h+2,0)})  \}  \\
&& \ \  \geq \delta_h + 3 \\
& & \ \  \geq  \frac{n-1+2m}{m}
\end{eqnarray*}
and
$$
 \dim(\bC(\bC_{(2,n, \delta_h+2,1)},  \bC_{(2,n, \delta_h+2,0)}))=2^{m-1}+2^{m-1}-1=n.
$$
Furthermore, the code $\bC(\bC_{(2,n, \delta_h+2,1)},  \bC_{(2,n, \delta_h+2,0)})$
has generator polynomial $(x-1)g(x)^2$.
This completes the proof.
}
\end{proof}

The code $\bC_{(2,n, \delta_h+2,1)}$ was proposed and studied in \cite{SunLiDing}, where
a much better lower bound on $d((\bC_{(2,n, \delta_h+2,1)})$ was developed.  Specifically,
the following result on $\delta_h$ was proved in  \cite{SunLiDing}:
\begin{itemize}
\item When $m \in \{, 3,5\}$,  $\delta_h=\frac{2^m-2}{m}-1$.
\item When $m=7$,  $\delta_h=\frac{2^m-2}{m}+1$.
\item When $m \geq 11$ is a prime,  $\delta_h \geq  \frac{2^m-2}{m}-1+2\Lambda $, where
\begin{align*}
\Lambda=\sum_{i=2}^{\lfloor\frac{m}{\upsilon(m)+2}\rfloor} \frac{(-1)^{i}}{i}\binom{m-i(\upsilon(m)+1)-1}{i-1}2^{m-i(\upsilon(m)+2)}
+\left\lfloor \frac{2^{m-\upsilon(m)-2}-2^{\upsilon(m)}}{2^{\upsilon(m)+1}-1} \right\rfloor+1, 	
\end{align*}
and $\upsilon(m)= \lfloor \log_2 m \rfloor$.
\end{itemize}
Using the better lower bound on $\delta_h$ above,  the lower bound in \eqref{eqn-sldbound} on
$$ d(\bC(\bC_{(2,n, \delta_h+2,1)},  \bC_{(2,n, \delta_h+2,0)}))$$ can be improved to a large extent
for $m \geq 11$.

\begin{example}
When $m=3$,  then the cyclic code $\bC(\bC_{(2,n, \delta_h+2,1)},  \bC_{(2,n, \delta_h+2,0)})$ has parameters
$[14,7,4]_2$.
This code is distance-optimal according to \cite{Grassl}, where the linear code with the same parameters is not known
to be cyclic.
\end{example}

\begin{example}
When $m=5$,  then the cyclic code $\bC(\bC_{(2,n, \delta_h+2,1)},  \bC_{(2,n, \delta_h+2,0)})$ has parameters
$[62,31,8 ]_2$.
The best linear code with parameters $[62,31,12]$ in \cite{Grassl} is not known to be cyclic.
 \end{example}

The following theorem presents another infinite family of repeated-root binary cyclic codes.

\begin{theorem}\label{thm-new222}
Let $m \geq 3$ be a prime and $n=2^m-1$. Put
$$
h_2=\frac{2^{m-1}-1}{m} \mbox{ and } h_1=\frac{h_2+1}{2}=\frac{2^{m-1}-1 +m}{2m}.
$$
Let $\delta_{h_i}$  denote the
$h_i$-th largest nonzero $2$-cyclotomic coset leader modulo $n$.
Then the cyclic code $\bC(\bC_{(2,n, \delta_{h_1},1)},  \bC_{(2,n, \delta_{h_2}+2,1)})$
has parameters $[2n, n+2^{m-2} + (m+1)/2, d]_2$, where
\begin{equation}\label{eqn-sldbound2}
d  \geq \frac{n-1}{m}.
\end{equation}
\end{theorem}

\begin{proof}
{\em
By definition, each nonzero $2$-cyclotomic coset leader modulo $n$ must be a positive odd integer. Consequently,
$$
\delta_{h_2} \geq 2h_2-1=\frac{n-1}{m} - 1
$$
and
$$
\delta_{h_1} \geq 2h_1-1=\frac{n-1}{2m} .
$$
It follows from the BCH bound that
$$
d(\bC_{(2,n, \delta_{h_2}+2,1)}) \geq \delta_{h_2} +2 \geq \frac{n-1+m}{m}
$$
and
$$
d(\bC_{(2,n, \delta_{h_1},1)}) \geq \delta_{h_1}  \geq \frac{n-1}{2m}
$$
Since $m$ is a prime, each nonzero $2$-cyclotomic coset must have cardinality $m$. Consequently,
$$
\dim(\bC_{(2,n, \delta_{h_2}+2,1)})=2^{m-1}.
$$
and
$$
\dim(\bC_{(2,n, \delta_{h_1},1)})=n-(h_1-1)m =n- \frac{2^{m-1}-1-m}{2}.
$$
Since $h_1 = (h_2 +1)/2$,  $\bC_{(2,n, \delta_{h_2}+2,1)}$is a
subcode of $\bC_{(2,n, \delta_{h_1},1)}$.
Since $\gcd(2, n)=1$, $x^n-1$ has no repeated roots.
It then follows from Theorem \ref{thm-evanlint} that
\begin{eqnarray*}
&& \dim(\bC(\bC_{(2,n, \delta_{h_1},1)},  \bC_{(2,n, \delta_{h_2}+2,1)})) \\
&& \ \ =  \dim(\bC_{(2,n, \delta_{h_1},1)}) + \dim( \bC_{(2,n, \delta_{h_2}+2,1)}) \\
&& \ \ =n+2^{m-2} + (m+1)/2
\end{eqnarray*}
and
\begin{eqnarray*}
&& d(\bC(\bC_{(2,n, \delta_{h_1},1)},  \bC_{(2,n, \delta_{h_2}+2,1)})) \\
&& \ \ = \min\{ 2d(\bC_{(2,n, \delta_{h_1},1)}),  d(\bC_{(2,n, \delta_h+2,1)})  \}  \\
& & \ \  \geq  \frac{n-1}{m}.
\end{eqnarray*}
}
\end{proof}

\begin{example}
When $m=3$,  then the cyclic code $\bC(\bC_{(2,n, \delta_{h_1},1)},  \bC_{(2,n, \delta_{h_2}+2,1)})$ has parameters
$[14,11,2]_2$.
This code is distance-optimal according to \cite{Grassl}, where the linear code with the same parameters is not known
to be cyclic.
\end{example}

\begin{example}
When $m=5$,  then the cyclic code  $\bC(\bC_{(2,n, \delta_{h_1},1)},  \bC_{(2,n, \delta_{h_2}+2,1)})$ has parameters
$[62,42,6 ]_2$.
The best linear code with parameters $[62,42,8]$ in \cite{Grassl} is not known to be cyclic.
 \end{example}

 To the best knowledge of the authors, no infinite family of binary linear codes with the same lengths  and
 dimensions better than any of the two infinite families of cyclic cods is reported in the literature.

\section{Summary of contributions and concluding remarks}\label{sec-fin}

The contributions of this paper are summarized as follows.
\begin{itemize}
\item An infinite family of distance-optimal repeated-root cyclic codes with parameters
  $[2(q^m-1), 2(q^m-1)-m-2,  3]_q$ was constructed (see Theorem \ref{thm-20302}), where $m \geq 2$ and $q \geq 4$ is even.

\item An infinite family of distance-optimal repeated-root cyclic codes with parameters
  $[2(q^m-1), 2(q^m-1)-2m-2,  4]_q$  was constructed (see Theorem 3,3), when $m \geq 2$ and $q \geq 4$ is even.

\item An infinite family of repeated-root cyclic codes with parameters
  $[2n, 2n-m-2,  3 \leq d \leq 4]_q$  was constructed (see Theorem \ref{thm-20301p}), where $n=(q^m-1)/(q-1)$, $m \geq 2$ and $q \geq 4$ is even.
No infinite family of linear codes over $\bF_q$ with the same length and dimension but better minimum distance is reported in the literature.

\item An infinite family of repeated-root cyclic codes with parameters
  $[2n, 2n-2m-2,  4]_q$  was constructed (see Theorem \ref{thm-2040p}), where $n=(q^m-1)/(q-1)$, $m \geq 2$ and $q \geq 4$ is even.
No infinite family of linear codes over $\bF_q$ with the same length and dimension but better minimum distance is reported in the literature.

\item An infinite family of distance-optimal repeated-root cyclic codesover $\bF_4$  with parameters
$[2(2^{2m-1}-1), 2(2^{2m-1}-1)-2m-1, 4]_4$   was constructed (see Theorem \ref{thm-2040h}), where $m \geq 2$.

\item An infinite family of repeated-root cyclic codes with parameters
  $[2(2^m-1), 2(2^m-1)-3m-1,  6]_2$  was constructed (see Theorem \ref{thm-3160}), where $m \geq 3$.
No infinite family of linear codes over $\bF_2$ with the same length and dimension but better minimum distance is reported in the literature.

\item An infinite family of repeated-root cyclic codes with parameters
  $[2(2^m-1), 2(2^m-1)-4m-2,  8]_2$  was constructed (see Theorem \ref{thm-3080}), where $m \geq 5$.
No infinite family of linear codes over $\bF_2$ with the same length and dimension but better minimum distance is reported in the literature.

\item Three infinite families of binary repeated-root cyclic codes with minimum distances $8$ or at least $10$ were constructed (see Theorems \ref{thm-281}, \ref{thm-282} and \ref{thm-210}).  To the best knowledge of
the authors, no infinite family of binary cyclic codes with the same lengths and dimensions better than any of the three infinite families of the binary cyclic codes is known.

\item An infinite family of repeated-root cyclic codes with parameters
  $[2(q^m-1), k,  6]_q$ was constructed (see Theorem \ref{thm-2060}), where $m \geq 2$ and
\begin{itemize}
\item $k=2(q^m-1)-4m-2$  if $q=4$,  and
\item  $k=2(q^m-1)-5m-2$  if $q>4$ is even.
\end{itemize}
No infinite family of linear codes over $\bF_q$ with the same length and dimension but better minimum distance is reported in the literature.

\item An infinite family of repeated-root cyclic codes over $\bF_4$ with parameters
$[2(2^{2m-1}-1), 2(2^{2m-1}-1)-6m+2, 6]_4$   was constructed (see Theorem \ref{thm-3160h}), where $m \geq 2$.
No infinite family of linear codes over $\bF_4$ with the same length and dimension but better minimum distance is reported in the literature.

\item An infinite family of repeated-root cyclic codes with parameters
 $[2(2^m-1), 2^m-1, d]_2$    was constructed (see Theorem \ref{thm-new111}), where
  $$
d  \geq \frac{n-1+2m}{m}
  $$
and
$m \geq 3$ is a prime.
    The lower bound on the minimum distance is much better than the square-root bound.
No infinite family of linear codes over $\bF_2$ with the same length and dimension but better minimum distance is reported in the literature.

\item An infinite family of repeated-root cyclic codes with parameters
$$
[2(2^m-1), 2^m-1+2^{m-2} + (m+1)/2, d \geq (2^m-2)/m]_2
$$
    was constructed (see Theorem \ref{thm-new222}), where $m \geq 3$ is a prime.
    The lower bound on the minimum distance is much better than the square-root bound.
No infinite family of linear codes over $\bF_2$ with the same length and dimension but better minimum distance is reported in the literature.
\end{itemize}
In summary, three infinite families of distance-optimal repeated-root cyclic codes over $\bF_q$ for
even $q$ were constructed in this paper.
In addition,  27 repeated-root cyclic codes of length up to $254$ over $\bF_q$ for $q \in \{2, 4, 8\}$
with optimal parameters or best parameters known were obtained in this paper (see Table
\ref{tab:A-q-5-3}).  Notice that some of the binary codes in Table \ref{tab:A-q-5-3} may be the same
as those in \citep{HSZ,Lint1,Massey}. \\

The results of this paper demonstrate that there are infinite families of repeated-root cyclic codes
over $\bF_q$ with minimum distance 3 or 4 for each even $q$.  Several families of distance-optimal
repeated-root binary cyclic codes with minimum distance 4 and several families of distance-optimal
repeated-root $p$-ary cyclic codes with minimum distance 3 for odd prime $p$ were constructed in
\cite{HSZ}. However,
it seems difficult to construct an
infinite family of distance-optimal repeated-root cyclic codes over small finite fields with minimum
distance 6 or more.  This is also true for the construction of  distance-optimal linear codes over small finite fields.
Although the theory and practice of cyclic codes have been extensively developed, repeated-root cyclic codes are less studied and understood. 	Further research into repeated-root cyclic codes would be necessary and interesting.  \\

Finally,  we would point out that all the repeated-root cyclic codes constructed in this paper are built
on BCH cyclic codes. This may explain why they are either distance-optimal or have the best parameters known compared with other infinite families of linear codes.  Further information on BCH codes could be
found in \citep{Ding4,Ding1}. \\

\begin{longtable}{|l|l|l|}
\caption{\label{tab:A-q-5-3} Repeated-root cyclic codes in this paper.}\\ \hline
$q$&Cyclic code &Optimality \\ \hline
$2$ &$  [14,7,4]_2 $& Optimal \\ \hline
$2$ &$ [14,11,2]_2  $& Optimal \\ \hline
$2$ &$  [30,14,8]_2 $& Optimal \\ \hline
$2$ &$  [30,17,6]_2 $& Best known \\ \hline
$2$ &$ [30,24,4]_2  $& Optimal \\ \hline
$2$ &$ [62,40,8]_2   $& Optimal \\ \hline
$2$ &$  [62,46,6]_2 $& Best known \\ \hline
$2$ &$ [62,55,4]_2  $& Optimal \\ \hline
$2$ &$  [126,100,8]_2 $& Best known \\ \hline
$2$ &$ [126,107,7]_2  $& Best known \\ \hline
$2$ &$ [126,118,4]_2  $& Optimal \\ \hline
$2$ &$ [254,224,8]_2  $& Best known \\ \hline
$2$ &$ [254,232,6]_2  $& Best known \\ \hline
$2$ &$ [254, 245, 4]_2  $& Optimal \\ \hline
$4$ &$[10,6,4]_4     $& Optimal \\ \hline
$4$ &$[14,9,4]_4     $& Optimal \\ \hline
$4$ &$[30,20,6]_4   $ &Best known  \\ \hline
$4$ &$[30,24,4]_4   $& Optimal \\ \hline
$4$ &$[42,37,3]_4   $& Optimal \\ \hline
$4$ &$[62,55,4]_4   $& Optimal \\ \hline
$4$ &$ [126,112,6]_4  $& Best known \\ \hline
$4$ &$[126,118,4]_4   $& Optimal \\ \hline
$4$ &$[126,121,3]_4   $& Optimal \\ \hline
$4$ &$ [170,164,3]_4  $& Optimal \\ \hline
$4$ &$ [254,245,4]_4  $& Optimal \\ \hline
$8$ &$[18,14,4]_8$   & Optimal \\ \hline
$8$ & $[126,120,4]_8$   & Optimal  \\ \hline
$8$ &$[126,122,3]_8   $& Optimal \\ \hline
\end{longtable}


\begin{thebibliography}{10}
		
\bibitem{BGG}
A. Batoul, K. Guenda and Hulliver, Repeated-root isodual cyclic codes over finite fields,
in: S. El Hajji et al. (Eds.): C2SI 2015, LNCS 9084, pp. 119--132, 2015.		
		
\bibitem{BC1} R. C. Bose and D. K. Ray-Chaudhuri, On a class of error-correcting binary group codes, Inf. and Contr., vol. 3, pp. 68-79, 1960.

\bibitem{BC2} R. C. Bose and D. K. Ray-Chaudhuri, Further results on error-correcting binary group codes, Inf. and Contr., vol. 3, pp. 279-290, 1960.


 \bibitem{Massey} G. Castagnoli, J. L. Massey, P. A. Scholler and N. von Seemann, On repeated-root cyclic codes, IEEE Trans. Inf. Theory, vol. 37, no. 2, pp. 337-342, 1991.
		

\bibitem{ChenDing} H. Chen and C. Ding, Self-dual cyclic codes with square-root-like lower bounds on their Minimum distances, preprint,  2023.

		
		
		
		\bibitem{Ding4} C. Ding, Parameters of several classes of BCH codes, IEEE Transactions on Information Theory, vol. 61, no. 10, pp. 5322-5330, 2015.
		
\bibitem{Ding0} C. Ding,  Codes from difference sets, Singapore: World Scientific, 2015		

\bibitem{Ding5} C. Ding and T. Helleseth, Optimal ternary cyclic codes from monomials, IEEE Trans. Inf. Theory, vol. 59, no. 9, pp. 5898-5904, 2013.

         \bibitem{Ding1} C. Ding and C. Li, BCH cyclic codes, Discrete Mathematics,  vol. 347, 113918, May 2024.

\bibitem{Fuetal}
W. Fang, J. Wen and F.  Fu, A $q$-polynomial approach to constacyclic codes,
Finite Fields Appl., vol. 47, pp. 161--182, 2017. 		
		

\bibitem{Grassl} M. Grassl, Bounds on the minimum distance of linear codes and quantum codes,
Online available at http://www.codetables.de.
		
		
		
		
		\bibitem{HT} C. R. P. Hartmann and K. K. Tzeng, Generalizations of the BCH bound, Infomation and Control, vol. 20, pp.489-498, 1972.	
		
		
\bibitem{Heng} Z. Heng, C. Ding and W. Wang, Optimal binary linear codes from maximal arcs, IEEE Trans. Inf. Theory, vol. 66, no. 9, pp. 5387-5394, 2020.


		


\bibitem{Hoc} A. Hocquenghem, Codes correcteurs d'erreurs, Chiffres (Paris), vol. 2, pp. 147-156, 1959.

\bibitem{HSZ}
S. Huang, Z. Sun and S. Zhu,  On the construction of several classes of optimal repeated-root cyclic codes,
ACTA ELECTRONICA SINICA, vol. 50, no. 1,  pp. 142--148, 2022.
		

		
		\bibitem{HP} W. C. Huffman and V. Pless, Fundamentals of error-correcting codes, Cambridge University Press, Cambridge, U. K., 2003.
		
		
		
\bibitem{LY}
X. Li and Q. Yue, The Hamming distances of repeated-root cyclic codes of length $5p^s$,
 Discrete Applied Mathematics, vol.  284, pp. 29--41, 2020. 		
		
		

		

         \bibitem{MScode} F. J.  MacWilliams and N. J. A. Sloane, The Theory of error-correcting codes, 3rd Edition, North-Holland Mathematical Library, vol. 16. North-Holland, Amsterdam, 1977.
		
	
         \bibitem{LDL} C. Li, C. Ding and S. Li, LCD cyclic codes over finite fields, IEEE Trans. Inf. Theory, vol. 63, no. 7, pp. 4344-4356, 2017.

          \bibitem{Willems} C. Martinez-Perez and W. Willems, Self-dual doubly even $2$-quasi-cyclic transitive codes are asymptotically good, IEEE Trans. Inf. Theory, vol. 53, no. 11, pp. 4302-4307, 2007.

\bibitem{PZL}
B. Pang, S. Zhu and J. Li,  On LCD repeated-root cyclic codes over finite fields,
J. Appl. Math. Comput., vol.  56, pp. 625--635, 2018.

     \bibitem{Plotkin} M. Plotkin, Binary codes with specified minimum distance, IRE Trans., vol. IT-6, pp. 445--450, 1960.


\bibitem{Prange} E. Prange, Cyclic error-correcting codes in two symbols, TN-57-013, Technical notes issued by Air Force Cambridge Research Labs, 1957.

		
		
		
		
		
		\bibitem{Roos} C. Roos, A new lower bound for the minimum distance of a cyclic codes,
		IEEE Trans.  Inf. Theory, vol. 29, no. 2, pp. 330-332, 1983.
		
		
		
\bibitem{Sobani}
R. Sobhani, 	Matrix-product structure of repeated-root cyclic codes over finite fields,
Finite Fields and Their Applications, vol. 39,  pp. 216--232, 2016.
	





    \bibitem{SunLiDing} Z. Sun, C. Li and C. Ding, An infinite family of binary cyclic codes with best parameters, IEEE Trans. Inf. Theory, DOI: 10.1109/TIT.2023.3307732.

\bibitem{TangDing} C. Tang and C. Ding, Binary $[n, \frac{n+1}{2}]$ cyclic codes with good minimum distances, IEEE Trans. Inf. Theory, vol. 68, no. 12, pp. 7842-7849, 2022.


\bibitem{Lint1} J. H. van Lint, Repeated-root cyclic codes, IEEE Trans. Inf. Theory, vol. 37, no. 2, pp. 343-345, 1991.

\bibitem{Lint} J. H. van Lint, Introduction to coding theory, Springer, 3rd Edition, Berlin, Hong Kong and Tokyo, 1999.
		

\bibitem{Wang} X. Wang, D. Zheng and C. Ding, Some punctured codes of several families of binary linear codes, IEEE Trans. Inf. Theory, vol. 67, no. 8, pp. 5133-5148, 2021.
		
		
\bibitem{ZU}
A. Zeh and M. Ulmschneider	, Decoding of repeated-root cyclic codes up to new bounds on their minimum distance,  	Probl.  Inf.  Transm.  vol. 51, no. 3, pp. 217--230,  2015 (arXiv:1506.02820 [cs.IT]).


	\end{thebibliography}
\end{document}